\documentclass[sigconf,screen]{acmart}

\usepackage{geometry}
\usepackage[parfill]{parskip}
\usepackage{graphicx}
\usepackage{amssymb}
\usepackage{epstopdf}
\usepackage[utf8]{inputenc}
\usepackage[english]{babel}
\usepackage[T1]{fontenc}
\usepackage{caption}
\usepackage{subcaption}
\usepackage{amsmath}
\usepackage{proof}
\usepackage[]{algorithm2e}
\usepackage{bm}
\usepackage{listings}
\usepackage{fp}
\usepackage{xcolor}
\usepackage{colortbl}
\usepackage{booktabs}
\usepackage{multirow}
\usepackage{balance}

\setcopyright{acmlicensed}

\newcommand{\term}[1]	{\emph{\textcolor{blue}{#1}}}
\newcommand{\id}[1]		{\mathrm{#1}}

\newcommand{\basecase}	{\textbf{Base Case:~}}
\newcommand{\indhypo}	{\textbf{Induction Hypothesis:~}}
\newcommand{\indstep}	{\textbf{Induction Step:~}}

\newtheorem{theorem}{Theorem}
\newtheorem{definition}{Definition}

\newcommand{\BD}{\begin{definition}}
\newcommand{\ED}{\end{definition}}

\DeclareMathOperator{\Prop}{Prop}
\DeclareMathOperator{\Conf}{Conf}
\DeclareMathOperator{\sat}{sat}

\newcommand{\true}  {True~}
\newcommand{\false}  {False~}
\newcommand{\ldef}[1]{{\color{blue} #1}}
\newcommand{\souffle}{Souffl\'{e}}

\begin{document}
    
\copyrightyear{2019} 
\acmYear{2019} 
\acmConference[ESEC/FSE '19]{Proceedings of the 27th ACM Joint European Software Engineering Conference and Symposium on the Foundations of Software Engineering}{August 26--30, 2019}{Tallinn, Estonia}
\acmBooktitle{Proceedings of the 27th ACM Joint European Software Engineering Conference and Symposium on the Foundations of Software Engineering (ESEC/FSE '19), August 26--30, 2019, Tallinn, Estonia}
\acmPrice{15.00}
\acmDOI{10.1145/3338906.3338928}
\acmISBN{978-1-4503-5572-8/19/08}

\title{Lifting Datalog-Based Analyses to Software Product Lines}

\author{Ramy Shahin}
\email{rshahin@cs.toronto.edu}
\affiliation{
    \institution{University of Toronto}
    \country{Canada}
}
\author{Marsha Chechik}
\email{chechik@cs.toronto.edu}
\affiliation{
    \institution{University of Toronto}
    \country{Canada}
}
\author{Rick Salay}
\email{rsalay@cs.toronto.edu}
\affiliation{
    \institution{University of Toronto}
    \country{Canada}
}

\begin{abstract}
Applying program analyses to Software Product Lines (SPLs) has been a fundamental research problem at the intersection of Product Line Engineering and software analysis. Different attempts have been made to "lift" particular product-level analyses to run on the entire product line. In this paper, we tackle the class of Datalog-based analyses (e.g., pointer and taint analyses), study the theoretical aspects of lifting Datalog inference, and implement a lifted inference algorithm inside the \souffle~Datalog engine. We evaluate our implementation on a set of benchmark product lines. We show significant savings in processing time and fact database size (billions of times faster on one of the benchmarks) compared to brute-force analysis of each product individually. 
\end{abstract}




\maketitle

\newcommand{\resolveLifted}{\id{\widehat{infer}}}

\section{Introduction}
Software Product Lines (SPLs) are families of related products, usually developed together from a common set of artifacts. Each product configuration is a combination of features. As a result, the number of potential products is combinatorial in the number of features.
 This high level of configurability is usually desired. However, analysis tools (syntax analyzers, type checkers, model checkers, static analysis tools, etc...) typically work on a single product, not the whole SPL. Applying an analysis to each product separately is usually infeasible for non-trivial SPLs because of the exponential number of products~\cite{Liebig:2013}.

Since all products of an SPL share a common set of artifacts, analyzing each product individually (usually referred to as \term{brute-force analysis}) would involve a lot of redundancy. How to leverage this commonality and analyze the whole product line at once, bringing the total analysis time down, is a fundamental research problem at the intersection of Product Line Engineering and software analysis. Different attempts have been made to \term{lift} individual analyses to run on product lines~\cite{Bodden:2013, Classen:2010, Gazzillo:2012, Kastner:2011, Kastner:2012, Midtgaard:2015, Salay:2014}. Those attempts show significant time savings when the SPL is analyzed as a whole compared to brute-force analysis. The downside though is the amount of effort required to correctly lift each of those analyses.

In this paper, we tackle the class of Datalog-based program analyses. Datalog is a  declarative query language that adds logical inference to relational queries. Some program analyses (in particular, pointer and taint analyses) can be fully specified as sets of Datalog inference rules. Those rules are applied by an inference engine to facts extracted from a software product. Results are more facts, inferred by the engine based on the rules. The advantage of Datalog-based analyses is that they are declarative, concise and can be efficiently executed by highly optimized Datalog engines~\cite{Jordan:2016, Lhotak:2008}.

Instead of lifting individual Datalog-based analyses, we lift a Datalog engine. This way any analysis running on the lifted engine is lifted for free. Our approach is not specific to a particular engine though, and can be implemented in others. 
%Our evaluation of a few pointer analyses running on a lifted Datalog engine show significant time and space savings (several orders of magnitude) compared to brute-force analysis.  

{\bf Contributions}
%\subsection{Contributions}
In this paper we make the following contributions:
(1) We present {$\resolveLifted$}, a Datalog inference algorithm lifted to facts extracted from Software Product Lines.
(2) We state the correctness criteria of lifted Datalog inference and show that {$\resolveLifted$} is correct. %, and we outline a proof of the correctness of our lifted inference algorithm.
(3) We implement our lifted algorithm as a part of a Datalog engine. We also extend the Doop pointer analysis framework~\cite{Bravenboer:2009} to extract facts from SPLs.
(4) We evaluate our implementation on a sample of pointer and taint analyses applied to a suite of Java benchmarks. We show significant savings in processing time and fact database sizes compared to brute-force analysis of one product at a time. For one of the benchmarks, our lifted implementation is billions of times faster than brute-force analysis (with savings in database size of the same order of magnitude).
%\end{enumerate}

The rest of the paper starts with a background on SPLs and 
Datalog (Sec.~\ref{sec:background}).
%this paper starts with some background on Software Product Lines, Datalog and Datalog-based software% analyses (Sec.~\ref{sec:background}). 
We provide a theoretical treatment of Datalog inference, how the inference algorithm is lifted, together with correctness criteria and a correctness proof in Sec.~\ref{sec:lifting}. In Sec.~\ref{sec:implementation}, we describe the implementation of our algorithm in the \souffle~engine. Evaluation process and results are discussed in Sec.~\ref{sec:evaluation}. We compare our approach to related work in Sec.~\ref{sec:related} and conclude (Sec.~\ref{sec:conclusion}).
 
\section{Background}
\label{sec:background}
In this section, we summarize
the basic concepts of Software Product Lines, Horn Clauses, Datalog and Datalog-based analyses.

\newcommand{\pc}[1]	{pc_{#1}}
\newcommand{\SPL}  	{\mathcal{L}}
\newcommand{\feat}[1]{\bm{#1}} %for features
\newcommand{\featset}{F}
\newcommand{\featmodel}{\Phi}
\newcommand{\domainmodel}{D}
\newcommand{\pcmap}{\phi}
\newcommand{\config}{\rho}
\newcommand{\FA}{\feat{A}}
\newcommand{\FB}{\feat{B}}

\subsection{Software Product Lines}
A \term{Software Product Line (SPL)} is a family of related software products developed together.
% from a common set of artifacts, in addition to feature-specific artifacts. 
Different variants of an SPL have different \term{features}, i.e., externally visible attributes such
as a piece of functionality, support for a particular peripheral device, or a performance optimization.

\BD[SPL]
An SPL $\SPL$ is a tuple $(\featset, \featmodel, \domainmodel, \pcmap)$ where:
(1) $\featset$ is the set of features s.t.  an individual product can be derived from $\SPL$ via a \term{feature configuration} $\config \subseteq \featset$.
(2) $\featmodel \in \Prop(\featset)$ is a propositional formula over $\featset$ defining the valid set of feature configurations. $\featmodel$ is called a \term{Feature Model (FM)}. The set of valid configurations defined by $\featmodel$ is called $\Conf(\SPL)$.
(3) $\domainmodel$ is a set of program elements, called the \term{domain model}. The whole set of program elements is sometimes referred to as the \term{150\% representation}. 
(4) $\pcmap:\domainmodel \to \Prop(\featset)$ is a total function mapping each program element to a proposition (\term{feature expression}) defined over the set of features $F$. $\pcmap(e)$ is called the \term{Presence Condition (PC)} of element $e$, i.e. the set of product configurations in which $e$ is present.
\ED

{\bf Example.}
Consider the \term{annotative} Java product line with feature set $F=\{FA, FB\}$, shown  in Listing~\ref{lst:SPL}.  Features are annotated using the \term{C Pre-Processor(CPP)} 
conditional compilation directives. By defining or not-defining macros corresponding to features, different products can be generated from this product line. 
One example is the product on Listing~\ref{lst:product}, with $FA$ not defined and $FB$ defined.

Here a single code-base (domain model $D$) is maintained, where different pieces of code are \term{annotated} with feature expressions. For example, tokens on line 10 are annotated with $\neg FA$.  That is, $\neg FA$
is the PC of these tokens. Similarly, tokens on line 13 have the PC $FB$.  This SPL allows all four
feature combinations, so its feature model $\featmodel$ is $\true$. 

\lstset{language=java}
\lstset{
    morekeywords={ifdef, endif}
}

%defined and FB not-defined. 
%\begin{figure*}[t]
%\centering
\begin{figure}[t!]
    \centering
    \begin{lstlisting}[numbers=left, captionpos=b, caption=A product line with features $FA$ and $FB$., label=lst:SPL, morekeywords={\#ifdef, \#endif}]
    class Parent {public Object f;}
    class ClassA extends Parent {}
    class ClassB extends Parent {}
    
    ClassA o1 = new ClassA();
    ClassB o2 = new ClassB();
    #ifdef FA
    Parent o3 = o1;
    #else
    Parent o3 = o2;
    #endif
    #ifdef FB
    o2.f = o1;
    #else
    o2.f = o2;
    #endif
    Object r = o3.f;
    \end{lstlisting}
\end{figure}
\begin{figure}[t!]
    \centering
    \begin{lstlisting}[language=java, numbers=left, captionpos=b, caption=A product with a configeration $(\neg FA \land FB)$., label=lst:product]
    Class Parent {public Object f;}
    Class ClassA extends Parent {}
    Class ClassB extends Parent {}
    
    ClassA o1 = new ClassA();
    ClassB o2 = new ClassB();
    Parent o3 = o2;
    o2.f = o1;
    Object r = o3.f;
    \end{lstlisting}
\end{figure}
%\caption{An annotative product line and  a product generated from it.}
%\label{fig:PL_example}
%\end{figure*}

\newcommand{\fact}[1]{f_{#1}}
\newcommand{\premise}[1]{s_{#1}}
\newcommand{\colonminus}{\text{:-}}
\newcommand{\assignment}{\gamma}
\newcommand{\IDB}{IDB}
\newcommand{\EDB}{EDB}
\newcommand{\liftedEDB}{\widehat{\EDB}}

\subsection{Horn Clauses and Datalog}

\subsubsection{Horn Clauses}
A \term{Horn Clause (HC)} is a disjunction of unique propositional literals with at most one positive literal. For example, $(\neg a \vee \neg b \vee c \vee \neg d)$ is an HC which can be also written as a reverse-implication $(c \leftarrow (a \wedge b \wedge d))$, where $c$ is called the \term{head} and $(a \wedge b \wedge d)$ is called the \term{body} of the clause. 
The language of HCs 
%(named after Alfred Horn~\cite{Horn:1951}) 
is a fragment of Propositional Logic that can be checked for satisfiability in linear time, as opposed to general propositional satisfiability which is NP-complete~\cite{Huth:2004}. 
\subsubsection{Datalog}
Datalog is a declarative database query language that extends relational algebra with logical inference~\cite{Greco:2016}. Datalog inference \term{rules} are HCs in First Order Logic, where atoms are predicate expressions, not just propositional literals. A \term{fact} is a ground rule with only a head and no body. 
Syntactically, the ':-' symbol is usually used instead of backward implication, and atoms in the body are separated by commas
instead of the conjunction symbol.
%Datalog rules are usually written in a slightly different notation than HCs. Instead of the backward implication sign, ':-' is used. Also atoms of the body are separated by commas instead of the conjunction symbol. This is just a convenient syntactic variation with no semantic significance.

Fig.~\ref{fig:HC_grammar} defines the grammar of Datalog clauses as follows: (1)
    building blocks are finite sets of constants, variables and predicate symbols;
    (2) a \emph{term} is a constant or a variable symbol;
    (3) a predicate expression is an $n$-ary predicate applied to arguments;% taking $n$ terms as arguments.
    (4) a \term{fact} is a ground predicate expression, i.e., all of its arguments are constants;
    (5) a \term{rule} is a Horn Clause of predicate expressions; and (6)
    a \emph{Datalog clause} is either a fact or a rule.

A Datalog program is a finite set of rules, usually referred to as the \term{Intensional Database (IDB)}, which operates on a finite set of facts called the \term{Extensional Database (EDB)}. The inference algorithm (explained next) repeatedly applies the rules to the facts, inferring new facts and adding them to the EDB, until a fixed point is reached (i.e., no more new facts can be inferred).

%% generative grammars for HCs and lifted HCs
\begin{figure*}[t]
    \begin{tabular}{ll}
    % HCs
    \begin{subfigure}[c]{0.49\textwidth}
        \centering
        \[
        \begin{array}{lcl}
        S &  = & \{\text{finite set of constant symbols}\} 	\\
        V &  = & \{\text{finite set of variables}\}			\\
        P &  = & \{\text{finite set of predicate symbols}\}	  \\
        T & ::= & S~|~V \\
        L & ::= & P(T_1, ..., T_n)	\\
        F & ::= & P(S_1, ..., S_m)    \\
        R & ::= & L_0~\colonminus~L_1, ..., L_k \\
        D & ::= & F~|~R
        \end{array}
        \]
        \caption{Datalog Grammar.}
        \label{fig:HC_grammar}
    \end{subfigure}
    \begin{tabular}{l}
    % variable assignment
    %lifted HCs
    \begin{subfigure}[c]{0.49\textwidth}
         \centering
        $\assignment : V \to S$ \\
        $[\assignment]C = C[v/\assignment(v)]$ , for each free variable $v$ in $C$ \\ 
       \caption{Variable assignment function and substitution for clause $C$.}
        \label{fig:HC_subst}
    \end{subfigure}
    \\
    \\
    \begin{subfigure}[c]{0.49\textwidth}
        \centering
        \[
        \begin{array}{lcl}
        PC & ::= & f~|~\neg PC~|~PC \land PC~|~PC \lor PC \\
        \widehat{D} & ::= & (F, PC)~|~R 
        \end{array}
        \]
        \caption{Grammar for lifted Datalog clauses. Syntactic category $f$ is set of feature names.}
        \label{fig:HC_lifted_grammar}
    \end{subfigure}
    %lifted HCs
\end{tabular}
\end{tabular}
    \caption{(a) Grammar of Datalog clauses, (b) variable assignment function and substitution, and (c) lifted Datalog clauses.} 
%   Figure ugly and needs compaction.  It is taking too much space} \CMNT{Tried to clean up a little, but I% think it still takes same amount of space.}
\end{figure*}

\subsubsection{Inference Algorithm (Algorithm~\ref{alg:inference})}
%Starting with a set of facts stored in an EDB and a program in the form of IDB rules, the inference algorithm looks for rule applications, and adds their results to the EDB. 
For each rule $R$, the algorithm checks to see if the EDB has facts fulfilling the premises of $R$, with a consistent assignment of variables to constants (Fig.~\ref{fig:HC_subst}). If it does, the head of that rule is inferred as a new fact $F$. If $F$ doesn't already exist in the EDB, it is added to it. Newly inferred facts may trigger some of the rules again; this process continues until a fixed point
is reached, i.e., no new facts are inferred. This algorithm (called the \term{forward chaining} algorithm~\cite{Ceri:1989})
%If new facts were inferred and added to the EDB, we need to see if they would trigger any of the rules again. This is repeated until we reach a fixed point, i.e. no more new facts are inferred.
is guaranteed to terminate because it does not create any new constants, and
runs in polynomial time w.r.t. the number of input clauses~\cite{Ceri:1989}.
%Algorithm~\ref{alg:inference} is called the \term{forward chaining} algorithm because it moves forward in the direction of implication for each of the rules (from the body to the head). Since no new constants are created in the process, this algorithm is guaranteed to terminate. Moreover, it has been proven that it runs in polynomial time with respect to the number of input clauses~\cite{Ceri:1989}.  

%% inference algorithms for HCs and lifted HCs
\begin{algorithm}
    \caption{Inference algorithm \emph{infer} (forward chaining).}
    \label{alg:inference}
    \SetAlgoLined
    \KwData{input: IDB, EDB}
    \KwResult{EDB + inferred clauses}
    \Repeat{fixpoint}
    {
        fixpoint = True\;
        \ForEach
        {$(C$ \colonminus~$\premise{1}, ..., \premise{n}) \in IDB$}
        {\ForEach{$(\assignment, \fact{1},...,\fact{n}), \fact{i} \in EDB, [\assignment]\premise{i} = \fact{i}$}
             {\If {$[\assignment]C \notin EDB$} 
                {fixpoint = False\; $EDB = EDB \cup \{[\assignment]C\}$}
             }
        }
    }
    return EDB\;
\end{algorithm}

\subsection{Example of a Datalog Analysis}
Some program analyses~\cite{Benton:2007, Bravenboer:2009, Dawson:1996, Grech:2017} can be written in Datalog as sets of clauses. Facts relevant to the analysis are extracted from the program to be analyzed, and then fed into a Datalog engine together with the analysis clauses.
%One class of program analyses is those written in the form of Datalog rules. Facts are first extracted from the program to be analyzed, and then fed into a Datalog engine together with the analysis rules. 
Fact extraction is usually analysis-specific because different analyses work on different aspects of the program. One example of Datalog-based analyses is pointer analysis.

%\subsubsection{Pointer Analysis}
\begin{figure*}[t]
%\centering
\begin{tabular}{ll}
\begin{subfigure}{0.49\textwidth}
%% rules
\begin{lstlisting}
VarPointsTo(v1, h1) :- New(v1, h1).
VarPointsTo(v1, h2) :- 
    Assign(v1, v2), VarPointsTo(v2, h2).
VarPointsTo(v1, h2) :- 
    Load(v1, v2, f),
    VarPointsTo(v2, h1),
    HeapPointsTo(h1, f, h2).
HeapPointsTo(h1, f, h2) :- 
    Store(v1, f, v2),
    VarPointsTo(v1, h1),
    VarPointsTo(v2, h2).
\end{lstlisting}
\caption{Pointer Analysis Rules.}
\label{fig:pointer_analysis_rules}
\end{subfigure}
%\vfill
\begin{tabular}{l}
%% facts
\begin{subfigure}{0.49\textwidth}
\centering
\begin{lstlisting}
New("o1", "A").		// line 5
New("o2", "B"). 	// line 6
Assign("o3", "o2"). 	// line 7
Store("o2", "f", "o1"). // line 8
Load("r", "o3", "f").   // line 9
\end{lstlisting}
\caption{Facts extracted from Listing~\ref{lst:product}.}
\label{fig:pointer_analysis_facts}
\end{subfigure} \\
\begin{subfigure}{0.49\textwidth}
\begin{lstlisting}
VarPointsTo("o1", "A").
VarPointsTo("o2", "B").
VarPointsTo("o3", "B").
HeapPointsTo("B", "f", "A").
VarPointsTo("r", "A").
\end{lstlisting}
\caption{Results of applying the rules to the extracted facts.}
\label{fig:pointer_analysis_results}
\end{subfigure}
\end{tabular}
\end{tabular}
\caption{(a) Context-insensitive pointer analysis rules (simplistic), (b)input facts, and (c) output facts for program in Listing~\ref{lst:product}.}
\end{figure*}
%\end{figure*}
\term{Pointer analysis}~\cite{Smaragdakis:2015} determines which objects might be pointed to by a particular program expression. This whole-program analysis is %It is usually
%
%by is a category of program analyses commonly formulated as Datalog rules. Pointer analysis tries to answer the question of which objects a pointer (reference in the terminology of some programming languages) might be pointing to at runtime. Usually it is 
%implemented as a whole program analysis (inter-procedural). It is 
\term{over-approximating} in the sense that it returns a set of objects that \emph{might} be pointed to by each pointer, possibly with false positives.
Fig.~\ref{fig:pointer_analysis_rules} shows a set of Datalog rules for a simple pointer analysis~\cite{Madsen:2016}.  
Each predicate defines a relation between different artifacts. For example, \term{VarPointsTo(v,h)} states that pointer $v$ might point to heap object $h$. The first three rules specify the conditions for this predicate to hold: either a new object is allocated and a pointer is initialized; a pointer that already points to an object is assigned to another pointer; or an object field points to a heap object, and that field is assigned to another pointer. The fourth rule states that assigning a value to an object field results in that field pointing to the same object as the right-hand-side of the assignment.
%whatever the RHS , that field now points to whatever the right-hand-side of the assignment was pointing to. {\bf MC:  last part unclear} \CMNT{Paraphrased and clarified.}

Fig.~\ref{fig:pointer_analysis_facts} shows the facts corresponding to the program in Listing~\ref{lst:product}. The first two are object allocation facts; the third is an assignment fact, and the fourth and the fifth are store and load facts, respectively.
Fig.~\ref{fig:pointer_analysis_results} is the results of running the Datalog inference algorithm on those rules and facts.
%{\bf MC:  put all listings into a figure, arranging them so as to minimize space} \CMNT{Done.}
The example in Fig.~\ref{fig:pointer_analysis_rules} is called a \term{context-insensitive} pointer analysis because it does not distinguish between different objects, call sites and types in a class hierarchy. More precise context-sensitive pointer analyses take different kinds of context into consideration.
%, resulting in more precise results. 
For example, a \term{1-call-site-sensitive} analysis
considers method call sites.
 %of methods into consideration as a part of the analysis. 
 A \term{1-object-sensitive} analysis (similarly, \term{1-type-sensitive}) includes object allocation sites (types of objects allocated) as part of the context. 
 %Each of these context sensitive analyses come with more Datalog rules, and sometimes requires more facts to be extracted.
%\subsubsection{Taint Analysis}
%\emph{Taint analysis}~\cite{Grech:2017} is a static information flow analysis that tracks the propagation of sensitive information through different parts of a program. The goal is to avoid the leakage of sensitive information to untrusted code modules. P/Taint~\cite{Grech:2017} is a taint analysis framework implemented together with pointer analysis as Datalog rules.

% name definitions
\newcommand{\resolve}{\id{infer}}

\newcommand{\features}{\id{features}}
\newcommand{\products}{\id{products}}
\newcommand{\indexProd}[2]{#1|_#2}

\section{Lifting Datalog}
\label{sec:lifting}
In this section, we present our approach to lifting Datalog abstract syntax and the Datalog inference algorithm. We also formally state the correctness criteria for lifted Datalog inference, and outline a correctness proof of our lifted algorithm.
% 
%{\bf MC : describe what the section is doing} \CMNT{Done.}

%% modus ponens for HCs and lifted HCs
\begin{figure*}[t]
    \centering
    %% HCs
    \begin{subfigure}[b]{0.9\textwidth}
        \[
        \infer[MP]{[\assignment]C}{
            C \colonminus~\premise{1}, ..., \premise{n}
            & [\assignment]\premise{1} = \fact{1}
            & ...
            & [\assignment]\premise{n} = \fact{n}
            & \forall(1 \leq i \leq n), \fact{i} \in EDB
            & \assignment:V \to S
        }
        \]
        \label{fig:MP}
        %\caption{Modus Ponens for Datalog clause inference.}
    \end{subfigure}
    \vfill
    %% lifted HCs
    \begin{subfigure}[b]{0.9\textwidth}
        \centering
        \[
        \infer[\widehat{MP}]{([\assignment]C, \ldef{pc_1 \wedge ... \wedge pc_n})} {%
            C \colonminus~\premise{1}, ..., \premise{n}
            & [\assignment]\premise{1} = \fact{1}
            & ...
            & [\assignment]\premise{n} = \fact{n}
            & \forall(1 \leq i \leq n), (\fact{i}, \ldef{\pc{i}}) \in \liftedEDB
            & \assignment:V \to S
        }
        \]
        \label{fig:MPlifted}
        %\caption{Modus Ponens for lifted Datalog clause inference}
    \end{subfigure}
    \vspace{-0.2in}
    \caption{Modus ponens for (a) Datalog clauses and (b) lifted Datalog clause inference.}
    \label{fig:MPall}
\end{figure*}

\subsection{Annotated Datalog Clauses}
When analyzing a single software product, an initial set of facts is extracted from product artifacts, and analysis rules are applied to those facts, eventually adding newly inferred facts to the initial set. In the case of SPLs, a fact might be valid only in a subset of products, and not necessarily the entire product space. We have to associate a representation of that subset with each of the extracted facts. Similar to SPL annotation techniques, a \term{Presence Condition (PC)} is a succinct representation that can be used to annotate facts.

Facts annotated with PCs are called \term{lifted facts}, and are stored in a lifted Extensional Database -- \term{$\liftedEDB$}. Given a feature expression $\config$, we define $\indexProd{\liftedEDB}{\config}$ to be the set of facts from $\liftedEDB$ which only exist in the product set defined by $\config$:
\[
\indexProd{\liftedEDB}{\config} = \{f~|~(f, \pc{}) \in \liftedEDB \land \sat(\pc{} \land \config)\}
\]

When the Datalog inference algorithm is applied to annotated facts, we have to take the PCs attached to facts into account. Whenever the inference algorithm generates a new fact, we need to associate a PC to it. If $\fact{new}$ is generated from premises $\fact{1}, \fact{2}, ..., \fact{n}$, with PCs $\pc{1},...,\pc{n}$, then $\pc{new}$ attached to $\fact{new}$ should be the conjunction of the input PCs, i.e., $\pc{1} \land ... \land \pc{n}$. Intuitively, $\pc{new}$ represents the set of products in which $\fact{new}$ exists, which is the intersection of the sets of products in which the premises exist.
    
To avoid having too many generated facts that are practically vacuous, we check $\pc{new}$ for satisfiability. If it isn't satisfiable, then its corresponding fact exists in the empty set of products, i.e., non-existent. Those facts can be safely removed from $\liftedEDB$, potentially improving the performance of inference.

%{\bf MC:  too much passive voice.  maybe the intro to this section will make this part read better.
%Think about a bit of refactoring} \CMNT{Tried to paraphrase some of the passive sentences.}
    
\subsection{Lifted Inference Algorithm}
Algorithm~\ref{alg:inference_lifted} takes a set of Datalog rules (IDB) and a set of annotated facts ($\liftedEDB$) as input, and returns all inferred clauses, annotated with their corresponding presence conditions. The structure of this algorithm is similar to that of Algorithm~\ref{alg:inference}, with the exception of conjoining the presence conditions of the facts used in inference, and assigning the conjunction as the presence condition of the result. There are four cases for $(c, \pc{c})$ to consider:
(1) if $\sat(\pc{c})$ is \false ($\pc{c}$ is not satisfiable), then this result is ignored because it doesn't exist in any valid product;
(2) if $(c, \pc{c}) \in \liftedEDB$, then this result is also ignored because it already exists for the same set of products;
(3) if $(c, \pc{d}) \in \liftedEDB$, where $\pc{d} \neq \pc{c}$, then $(c, \pc{d})$ is replaced with $(c,\pc{d} \lor \pc{c})$ in $\liftedEDB$. This means we are expanding the already existing set of products in which $c$ exists to also include the set denoted by $\pc{c}$;
(4) if $c$ doesn't exist at all in $\liftedEDB$, we add $(c, \pc{c})$ to it.
%\end{itemize} 
For example, when the lifted inference algorithm is applied to the rules in Fig.~\ref{fig:pointer_analysis_rules} and annotated facts in Fig.~\ref{fig:pointer_analysis_SPL_facts}, the result is the following:
\begin{lstlisting}
VarPointsTo("o1", "A")         @ True.
VarPointsTo("o2", "B")         @ True.
VarPointsTo("o3", "A")         @ FA.
VarPointsTo("o3", "B")         @ !FA.
HeapPointsTo("B", "f", "A")    @ FB.
HeapPointsTo("B", "f", "B")    @ !FB.
VarPointsTo("r", "A")          @ FA.
VarPointsTo("r", "B")          @ !FA.
\end{lstlisting} 

\begin{algorithm}
    \caption{Lifted inference algorithm $\resolveLifted$.}
    \label{alg:inference_lifted}
    \SetAlgoLined
    \KwData{input: $IDB$, $\liftedEDB$}
    \KwResult{$\liftedEDB$ + annotated inferred clauses}
    \Repeat{fixpoint}
    {
        fixpoint = \true\;
         \ForEach
        {$(C$ \colonminus~$\premise{1}, ..., \premise{n}) \in IDB$}
        {\ForEach{$(\assignment, (\fact{1}, \ldef{\pc{1}}),...,(\fact{n}, \ldef{\pc{n}})), (\fact{i}, \ldef{\pc{i}}) \in \liftedEDB, [\assignment]\premise{i} = \fact{i}$}
            {
                \ldef{$\pc{c} = \pc{1} \land ... \land \pc{n}$}\;
                \If{\ldef{sat($\pc{c}$)}}
                {\If{$([\assignment]C, \ldef{\pc{c}}) \notin \liftedEDB$}
                    {
                        fixpoint = \false\;
                        \If{\ldef{$\exists \pc{d}, ([\assignment]C, \pc{d}) \in \liftedEDB$}}
                        {\ldef{$\pc{c} = \pc{c} \lor \pc{d}$\;
                                $\liftedEDB = \liftedEDB - \{([\assignment]C, \pc{d})\}$}}
                        \ldef{$\liftedEDB = \liftedEDB \cup \{[\assignment]C, \pc{c})\}$}}
                }
            }
        }
    }
    return $\liftedEDB$\;
\end{algorithm}

\subsection{Correctness Criteria}
When applying the lifted inference algorithm $\resolveLifted$ to a set of rules $IDB$ and a set of annotated facts $\liftedEDB$, we expect the result to be exactly the union of the results of applying $\resolve$ to facts from each product individually. Moreover, each clause in the result of $\resolveLifted$ has to be properly annotated (i.e., its presence condition has to represent exactly the set of products having this clause in their un-lifted  analysis results).
%This is formally specified in Theorem~\ref{th:correctness}.
\begin{theorem}
    Given an SPL $\SPL=(\featset, \featmodel, \domainmodel, \pcmap)$, a set of rules $IDB$, and a set of lifted facts $\liftedEDB$ annotated with feature expressions over $\featset$:
    \[
        \forall (\config \in \Conf(\SPL)), \indexProd{\resolveLifted(\liftedEDB)}{\config} = \resolve(\indexProd{\liftedEDB}{\config}) 
    \]
\label{th:correctness}
\end{theorem}

%\vspace{-0.25in}
\begin{proof}
%Set equality is equivalent to a double-implication of mutual set containment, so we need to prove the implication in both directions:
\begin{itemize}
    \item $C \in \indexProd{\resolveLifted(\liftedEDB)}{\config} \implies C \in \resolve(\indexProd{\liftedEDB}{\config})$ \\
        By structural induction over the derivation tree of $C$: \\
        \basecase $(C, \pc{}) \in \liftedEDB$, where $\sat(\pc{} \land \config)$. Then $C \in \indexProd{\liftedEDB}{\config}$ (by definition of restriction operator). Since inputs are already included in the output of $\resolve$, $C \in \resolve(\indexProd{\liftedEDB}{\config})$. \\
        \indhypo Given a rule $R = C~\text{\colonminus}~\premise{1}, ..., \premise{n}$, and a variable assignment $\assignment$, 
        \[ \forall(1 \leq i \leq n):[\assignment]\premise{i} \in \indexProd{\resolveLifted(\liftedEDB)}{\config} \implies [\assignment]\premise{i} \in \resolve(\indexProd{\liftedEDB}{\config}) \]
        
        \indstep $C$ is derived by $\widehat{MP}$ (Fig.~\ref{fig:MPall}) from rule $R$. Since all the premises of $C$ are in $\resolve(\indexProd{\liftedEDB}{\config})$ (induction hypothesis), then so is $C$ ($MP$).
    \item $C \in \resolve(\indexProd{\liftedEDB}{\config}) \implies C \in \indexProd{\resolveLifted(\liftedEDB)}{\config}$ \\
        By structural induction over the derivation tree of $C$: \\
        \basecase Assume $(C, \pc{}) \in \liftedEDB$, for some $\pc{}$, where $\sat(\pc{} \land \config)$. Then $(C,pc) \in \resolveLifted(\liftedEDB)$ (input included in output of $\resolveLifted$). Since $\pc{} \land \config$ is satisfiable, then $C \in \indexProd{\resolveLifted(\liftedEDB)}{\config}$ (definition of restriction).\\
        \indhypo Given a rule $R = C~\text{\colonminus}~\premise{1}, ..., \premise{n}$, and a variable assignment $\assignment$, 
        \[ \forall(1 \leq i \leq n):[\assignment]\premise{i} \in \resolve(\indexProd{\liftedEDB}{\config}) \implies [\assignment]\premise{i} \in \indexProd{\resolveLifted(\liftedEDB)}{\config} \]
        \indstep $C$ is derived by $MP$ (Fig.~\ref{fig:MPall}) from rule $R$. Since all the premises of $C$ are in $\indexProd{\resolveLifted(\liftedEDB)}{\config}$ (induction hypothesis), then so is $C$ ($\widehat{MP}$).
\end{itemize}    
\end{proof}
\section{Implementation}
\label{sec:implementation}
In this section, we explain how we lift the Doop pointer and taint analysis framework, together with its underlying \souffle~Datalog engine.
%
%{\bf MC explain what is coming} \CMNT{Done.}
\subsection{Lifting Doop}
\begin{figure}[t]
    \centering
    \includegraphics[width=0.45\textwidth]{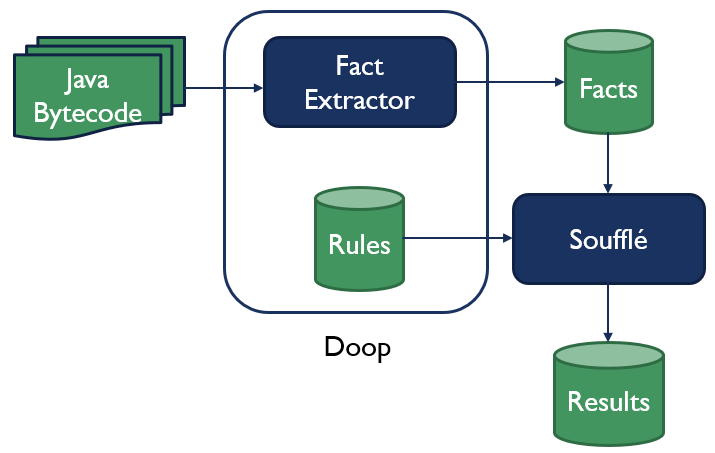}
    \caption{The Doop architecture.}
    \label{fig:doopArch}
\end{figure}
To illustrate and evaluate the Datalog lifting approach outlined in Sec.~\ref{sec:lifting}, we modified the Doop~\cite{Bravenboer:2009} Datalog-based pointer analysis framework \footnote{Available online at \url{https://bitbucket.org/rshahin/doop}}, together with its underlying \souffle~\cite{Jordan:2016} Datalog engine \footnote{Available online at \url{https://github.com/ramyshahin/souffle}}. Fig.~\ref{fig:doopArch} outlines the Doop architecture. Doop is an extensible family of pointer and taint analyses implemented as Datalog rules. In addition, it includes a fact extractor from Java bytecode. Doop users select a particular analysis among the available analyses through a command-line argument. The rules corresponding to the chosen analysis (the IDB), together with the extracted facts (the EDB), are then passed to \souffle.

% Here you want to discuss the modifications you have done rather then
%the Doop engine itself, right?  Or is it the next section?}
%
%\subsubsection{Lifting Fact Extraction}
%
%% facts
\begin{figure}[t]
\centering
\begin{lstlisting}
New("o1", "A")	    	@ True.	// line 5
New("o2", "B")	    	@ True. // line 6
Assign("o3", "o1")	@ FA.	// line 8
Assign("o3", "o2")	@ !FA. 	// line 10
Store("o2", "f", "o1")	@ FB. 	// line 13
Store("o2", "f", "o2")	@ !FB. 	// line 15
Load("r", "o3", "f")	@ True. // line 17
\end{lstlisting}
\caption{Annotated facts extracted from Listing~\ref{lst:SPL}.} 
\label{fig:pointer_analysis_SPL_facts}
\end{figure}
%{\bf MC:  wrong place for this.  Benchmark lives in the evaluation.
%Here you are just describing the implementation} \CMNT{Removed it from here.}

%The benchmarks we use for evaluation are annotated with CIDE~\cite{Kastner:2009}, an Eclipse-based framework for color-based disciplined annotations of code fragments. A given source code highlighting color represents a set of features. As a preprocessing step to feature extraction, we map color-based annotations to propositional feature expressions. We also map source code line numbers to their respective feature expressions. The benchmarks are compiled with symbol information included to trace byte code constructs  back to their respective source code line numbers. As a result, we can directly map each byte-code AST node to its feature expression. This way the feature extractor itself is totally independent of CIDE.
Since Doop extracts syntactic facts, we need to identify the PCs of each of the syntactic tokens contributing to a fact, and associate the conjunction of those PCs as the fact PC. We had to do this for each type of fact extracted by Doop. The fact PC is just added to a fact as a trailing PC field, prefixed with '@'. Facts with no PC field are assumed to belong to all products (an implicit PC of \true).

Our Doop modifications were only in the fact extractor. None of the Doop Datalog rules were changed. Our fact extraction modifications were scattered because extractors for different kinds of facts are implemented separately in Doop. However, all those changes were systematic and non-invasive. In total we modified only about 100 lines of code in the Doop fact extractor. 

\subsection{Lifting \souffle}
\begin{figure}[t]
    \centering
    \includegraphics[width=0.45\textwidth]{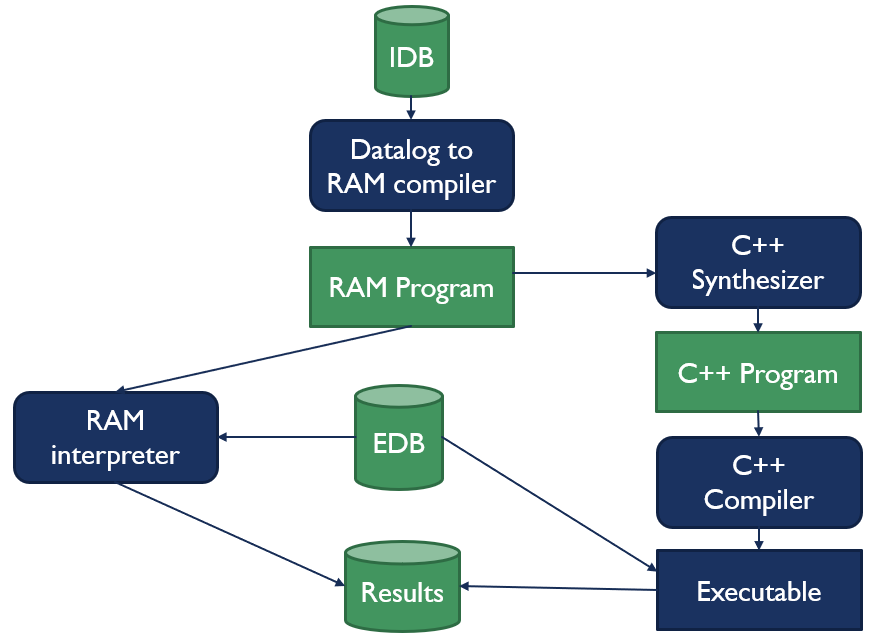}
    \caption{\souffle~architecture.}
    \label{fig:souffleArch}
\end{figure}
As seen in Fig.~\ref{fig:souffleArch}, a \souffle~program is first parsed and translated into a Relational Algebra Machine (RAM) program. RAM is a language with  relational algebra constructs, in addition to a fixed-point looping operator. Based on a command-line argument, \souffle~then either interprets the RAM program on the fly, or synthesizes C++ code that is semantically equivalent to the RAM program. Since C++ programs are compiled (typically by optimizing compilers) into native machine code, native executables are at least an order of magnitude faster than interpreted analyses~\cite{Jordan:2016}. In this paper, we only cover the \souffle~interpreter.
%
%\subsubsection{Lifting \souffle}

At the syntax level, we extend the \souffle~language with \term{fact annotations}. Those are propositional formulas prefixed with '@'. The \souffle~parser is extended with a syntactic category for propositional formulas. AST nodes for facts are extended with a PC field, with a default value of \term{True}. Propositional variables are added to a symbol table separate from that holding \souffle~identifiers.

As a part of compiling \souffle~programs into RAM, we turn syntactic presence conditions into Binary Decision Diagrams (BDDs). We use CUDD~\cite{Somenzi:1998} as a BDD engine, and on top of it maintain a map from textual presence conditions to their corresponding canonical BDDs. As stated in $\resolveLifted$, when facts are resolved with a rule, the conjunction of their PCs becomes the conclusion's PC. 

\souffle~implements several indexing and query optimization techniques to improve inference time. To keep our changes independent of those optimizations, we add the presence condition as a field opaque to the query engine. We only manipulate this field as a PC when performing clause resolution, which takes place at a higher level than the details of indexing and query processing. This way we avoid touching relatively complex optimization code, while preserving the semantics of our lifted inference algorithm.

Some relational features of \souffle~were not lifted. For example, aggregation functions (sum, average, max, min, etc...) still return singleton values. None of those functions is used by Doop on lifted facts, so this does not affect the correctness of our results. We still plan to address this general limitation in the future though.

\newcommand{\insens}{\term{insens}}
\newcommand{\TypeHeap}{\term{1Type+Heap}}
\newcommand{\CallHeapTaint}{\term{Taint-1Call+Heap}}
\newcommand{\RQ}[1]{\textbf{RQ#1}}

\section{Evaluation}
\label{sec:evaluation}
\begin{table}[htbp]
\caption{Java product lines used for evaluation.}
\label{tbl:benchmarks}
\vspace{-0.1in}
\begin{tabular}{lccc}
    \toprule
    Benchmark & Size (KLOC) & Features & Valid Configurations \\
    \midrule
    BerkeleyDB & 70 & 42 & 8,759,844,864 \\
    GPL & 1.4 & 21 & 4,176 \\
    Lampiro & 45 & 18 & 2,048 \\
    MM08 & 5.7 & 27 & 784 \\
    Prevayler & 7.9 & 5 & 32 \\
    \bottomrule
\end{tabular}
\end{table}
\vspace{-0.1in}
We evaluate the performance of our lifted version of Doop (together with lifted \souffle) on five Java benchmark product lines (previously used in the evaluation of other lifted analyses~\cite{Brabrand:2012,Bodden:2013}). For each of the benchmarks, Table~\ref{tbl:benchmarks} lists its size (in thousands of lines of code), number of features, and number of  valid configurations according to its feature model. For example, BerkeleyDB is about 70,000 lines of code, comprised of 42 features, and has about 8.76 billion valid product configurations.  We evaluate three Doop analyses: context-insensitive pointer analysis (\insens), one-type heap-sensitive pointer analysis (\TypeHeap), and one-call-site heap-sensitive taint analysis (\CallHeapTaint). For taint analysis, we use the default sources, sinks, transform and sanitization functions curated in Doop for the JDK and Android~\cite{Grech:2017}. All experiments were performed on a Quad-core Intel Core i7-6700 processor running at 3.4GHZ, with 16GB RAM and hyper-threading enabled, running 64-bit Ubuntu Linux (kernel version 4.15). 

Pointer and taint analyses work on the whole program, including library dependencies. Since general-purpose libraries usually do not have any variability, the comparison between lifted and single-product analyses is independent of them. Moreover, time spent in analyzing library code, and space taken by their facts, might skew the overall results. We restrict our experiments to application code and direct dependencies only using the Doop command-line argument "{\tt --Xfacts-subset APP\_N\_DEPS}".

Doop extracts its facts from Java byte-code. However, SPL annotation techniques work at the source-code level. Feature selection usually takes place at compile-time, which means an SPL codebase is compiled into a single product. To get around this limitation, we had to choose benchmarks that only have \term{disciplined annotations}~\cite{Kastner:2009}, in the sense that adding or removing an annotation preserves the syntactic correctness of the 150\% representation. This is not a limitation of our lifted inference algorithm though.  

The benchmarks we chose are annotated using CIDE~\cite{Kastner:2009}, which uses different highlighting colors as presence conditions. We had to extract this color information from CIDE, together with the mapping from colors to locations of tokens (line and column number) in source files. Our fact extractor uses byte-code symbol information to locate tokens, and assign their presence conditions based on CIDE colors.

The primary goal of our experiments is to compare the performance of lifted analyses applied to the SPLs to that of running the corresponding product-level analyses on each of the valid configurations individually. Since the number of valid product configurations for some benchmarks is relatively big, it is neither practical nor particularly useful to enumerate all of the valid products and analyze them. Instead, for each SPL, we run the product-level analysis on two code-base subsets: the base code common across all variants, and the 150\% representation (the whole SPL code-base, implementing all feature behaviors). Although those two extremes are not necessarily valid products, they are the lower bound and the upper bound in terms of code size, and averaging over them gives an "average" valid product approximation. The \emph{expected brute-force performance} is the average valid product performance (\term{P-Avg}) multiplied by the number of valid configurations.

We split our evaluation into two parts: fact extraction and inference, and  evaluate performance in terms of both the processing time and space (size of the fact database in kilobytes(KB)). Our primary research questions are:

\RQ{1}: How do fact extraction time (and size of the extracted fact database) of lifted analyses compare to brute-force fact extraction?

\RQ{2}: How do the \souffle\; inference time, and the size of the inferred database, of lifted analyses compare to brute-force analysis?

\subsection{Fact Extraction}

% Table generated by Excel2LaTeX from sheet 'FactExtraction'
\begin{table*}[htbp]
  \caption{Fact extraction time (in ms) and DB size (in KB): Average Product (P-Avg) vs SPL for all three analyses.}
  \label{tbl:factExtraction}
  \vspace{-0.1in}
  \small
  \centering
    \begin{tabular}{lrrrrrrrrrrrr}
        \toprule  
        \multicolumn{1}{r}{} & \multicolumn{4}{c}{\textit{\textbf{insens}}} & \multicolumn{4}{c}{\textit{\textbf{1Type+Heap}}} & \multicolumn{4}{c}{\textit{\textbf{Taint-1Call+Heap}}} \\
        \cmidrule(r){2-13}    
        \multicolumn{1}{r}{} & \multicolumn{2}{c}{\textit{\textbf{Time(ms)}}} & \multicolumn{2}{c}{\textit{\textbf{DB(KB)}}} & \multicolumn{2}{c}{\textit{\textbf{Time(ms)}}} & \multicolumn{2}{c}{\textit{\textbf{DB(KB)}}} & \multicolumn{2}{c}{\textit{\textbf{Time(ms)}}} & \multicolumn{2}{c}{\textit{\textbf{DB(KB)}}} \\
        \cmidrule(r){2-13}   
        \multicolumn{1}{r}{} & \multicolumn{1}{l}{\textbf{P-Avg}} & \multicolumn{1}{l}{\textbf{SPL}} & \multicolumn{1}{l}{\textbf{P-Avg}} & \multicolumn{1}{l}{\textbf{SPL}} & \multicolumn{1}{l}{\textbf{P-Avg}} & \multicolumn{1}{l}{\textbf{SPL}} & \multicolumn{1}{l}{\textbf{P-Avg}} & \multicolumn{1}{l}{\textbf{SPL}} & \multicolumn{1}{l}{\textbf{P-Avg}} & \multicolumn{1}{l}{\textbf{SPL}} & \multicolumn{1}{l}{\textbf{P-Avg}} & \multicolumn{1}{l}{\textbf{SPL}} \\
        \midrule
        \textbf{BerkeleyDB} & 5,136 & 4,541 & 31,892  & 49,725 & 4,651 & 4,809 & 71,536  & 122,922 & 4,647 & 4,667 & 64,497  & 112,060 \\
        %\hline
        \textbf{GPL} & 816   & 814   & 175   & 409   & 782   & 876   & 245   & 593   & 789   & 802   & 188   & 462 \\
        %\hline
        \textbf{Lampiro} & 4,413 & 4,475 & 41,100  & 41,170 & 4,425 & 4,145 & 149,521  & 149,686 & 4,237 & 4,436 & 230,035  & 230,370 \\
        %\hline
        \textbf{MM08} & 1,226 & 1,364 & 1,921  & 3,259 & 1,250 & 1,372 & 3,878  & 6,990 & 1,255 & 1,252 & 4,234  & 7,829 \\
        %\hline
        \textbf{Prevayler} & 1,416 & 1,554 & 3,230  & 4,407 & 1,453 & 1,454 & 5,882  & 8,630 & 1,502 & 1,404 & 3,917  & 5,534 \\
        \bottomrule
    \end{tabular}%
\end{table*}%

\begin{figure}[t]
\includegraphics[width=0.5\textwidth]{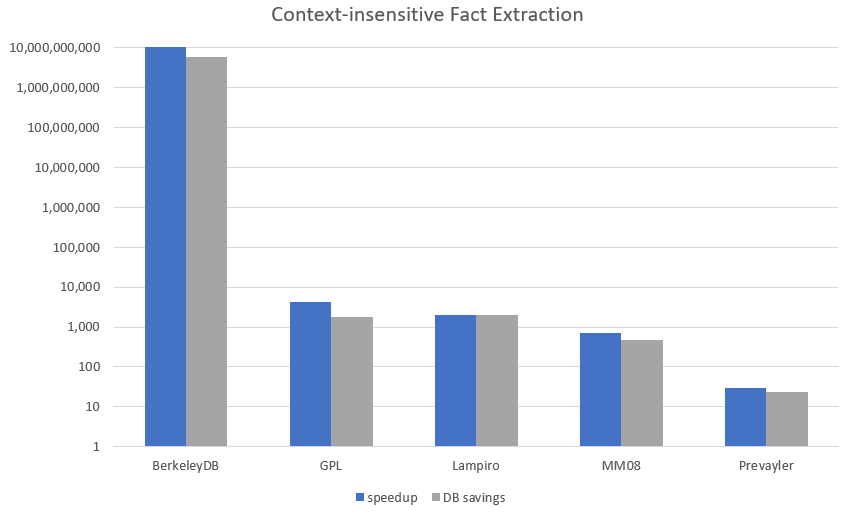}
\caption{Context-insensitive Fact extraction speedup and DB savings factors: SPL vs average product.}
\label{fig:graph_fact}
\vspace{-0.1in}
\end{figure}

Table~\ref{tbl:factExtraction} summarizes the "average" performance of product-level fact
extraction (\term{P-Avg}) and that of the lifted fact extraction for the entire product line (\term{SPL}). For each of the three analyses, we compare fact extraction time (in milliseconds) and the size of the extracted database (in KB).  For example, for context-insensitive analysis, average fact extraction time of a single product of Prevayler is 1,416ms, and the average size of the extracted fact database is 3,230KB. On the other hand, extracting facts from the whole Prevayler SPL at once takes 1,554ms, and the extracted fact database is 4,407KB.  
The difference between P-Avg Time and SPL Time is very small for all three analyses and five benchmarks, which is expected since extraction is syntactic and thus its time is proportional to code-base size, not the number of features.
 %According to the numbers, extraction time is very close for all analyses and benchmarks. This makes sense because we are extracting syntactic facts, which are expected to be proportional to the size of the code-base, not the number of features. 
 Size of the extracted database is noticeably bigger for 
 lifted extraction (DB SPL columns) because lifted facts are augmented with presence conditions.
 %though because we are now adding presence conditions as extra fields to facts that only belong to a proper subset of products.

To evaluate the savings attributed to lifted fact extraction compared to brute-force extraction in terms of time and space, we compute the speedup and space saving factors (P-Avg * $|\Conf(\SPL)|$ / SPL).
Fig.~\ref{fig:graph_fact} shows a log-scale bar graph of lifted fact extraction speedup and space savings for context-insensitive analysis. The other two analyses exhibit a similar trend and are omitted here. 
%show almost exactly the same trend, so we are showing the graph of only one analysis to enhance readability. 
The figure shows that the time and space savings are proportional to the number of valid configurations of the product line. For example, Lampiro has 2048 valid configurations, and its lifted fact extraction is 2020 times faster than brute-force, with a database 2045 times smaller than the total space of brute-force databases. On the other hand, Prevayler has only 32 valid configurations, with an \emph{insens} lifting speedup factor of 29, and a space savings factor of 23.
Since different analyses typically require different facts, the size of the fact database also varies from one analysis to another. Experimental results do not show a direct correlation between an analysis and the size of its fact database. For example, in Lampiro, the \CallHeapTaint~databases are significantly bigger than those of \TypeHeap. BerkeleyDB, on the other hand, exhibits the opposite trend.
 
\subsection{Inference}

% Table generated by Excel2LaTeX from sheet 'Inference'
\begin{table*}[htbp]
  \caption{Inference time (in ms) and inferred DB size (in KB): Average Product (P-Avg) vs SPL for the three analyses.}
  \label{tbl:inference}%
  \vspace{-0.1in}
  \small
  \centering
    \begin{tabular}{lrrrrrrrrrrrr}
    \toprule
   \multicolumn{1}{r}{} & \multicolumn{4}{c}{\textit{\textbf{insens}}} & \multicolumn{4}{c}{\textit{\textbf{1Type+Heap}}} & \multicolumn{4}{c}{\textit{\textbf{Taint-1Call+Heap}}} \\
    \cmidrule(r){2-13}
        \multicolumn{1}{r}{} & \multicolumn{2}{c}{\textit{\textbf{Time(ms)}}} & \multicolumn{2}{c}{\textit{\textbf{DB(KB)}}} & \multicolumn{2}{c}{\textit{\textbf{Time(ms)}}} & \multicolumn{2}{c}{\textit{\textbf{DB(KB)}}} & \multicolumn{2}{c}{\textit{\textbf{Time(ms)}}} & \multicolumn{2}{c}{\textit{\textbf{DB(KB)}}} \\
    \cmidrule(r){2-13}
       \multicolumn{1}{r}{} & \multicolumn{1}{l}{\textbf{P-Avg}} & \multicolumn{1}{l}{\textbf{SPL}} & \multicolumn{1}{l}{\textbf{P-Avg}} & \multicolumn{1}{l}{\textbf{SPL}} & \multicolumn{1}{l}{\textbf{P-Avg}} & \multicolumn{1}{l}{\textbf{SPL}} & \multicolumn{1}{l}{\textbf{P-Avg}} & \multicolumn{1}{l}{\textbf{SPL}} & \multicolumn{1}{l}{\textbf{P-Avg}} & \multicolumn{1}{l}{\textbf{SPL}} & \multicolumn{1}{l}{\textbf{P-Avg}} & \multicolumn{1}{l}{\textbf{SPL}} \\
    \midrule
        \textbf{BerkeleyDB} & 9,184 & 10,810 & 141,728 & 200,655 & 13,598 & 17,273 & 318,349 & 483,206 & 17,479 & 21,474 & 285,473 & 443,737 \\
    
        \textbf{GPL} & 4,422 & 4,517 & 2,033 & 3,528 & 4,794 & 4,718 & 3,237 & 5,675 & 8,999 & 8,861 & 2,500 & 4,450 \\
    
        \textbf{Lampiro} & 8,264 & 8,111 & 245,933 & 246,285 & 21,372 & 20,725 & 980,689 & 981,549 & 44,365 & 45,996 & 1,393,038 & 1,394,826 \\
    
        \textbf{MM08} & 4,596 & 4,720 & 8,788 & 13,021 & 5,106 & 5,142 & 19,058 & 29,453 & 9,340 & 9,306 & 18,302 & 29,383 \\
    
        \textbf{Prevayler} & 4,908 & 5,334 & 15,856 & 19,808 & 5,852 & 6,013 & 29,605 & 38,747 & 9,785 & 9,717 & 20,611 & 26,279 \\
    \bottomrule
    \end{tabular}%
\end{table*}%

\begin{figure}[t]
    \includegraphics[width=0.5\textwidth]{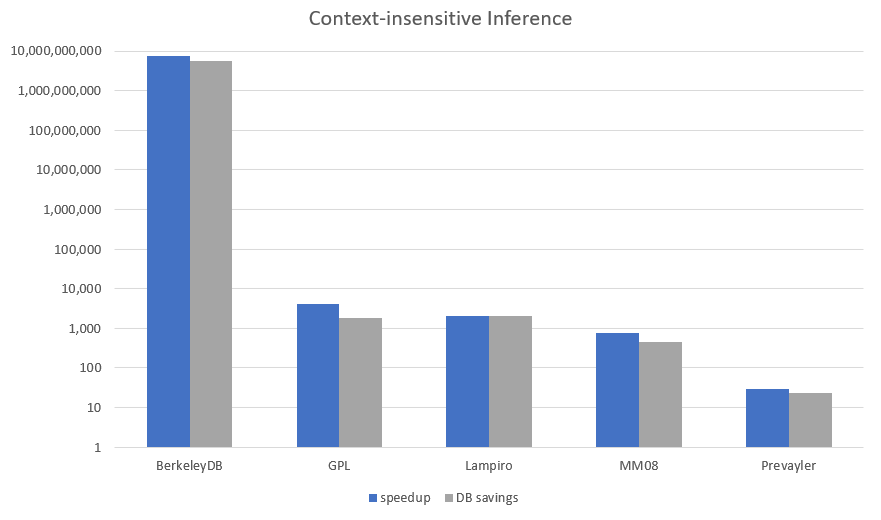}
    \caption{Context-insensitive Inference speedup and DB savings factors: SPL vs average product.}
    \label{fig:graph_inference}
    \vspace{-0.1in}
\end{figure}
Table~\ref{tbl:inference} summarizes the performance of lifted analyses on the entire product line (SPL) and that of product-level analyses on an average product (P-Avg). For example, running \TypeHeap~on an average MM08 product, inference is estimated to take 4,596ms, resulting in a database of 8,788KB. Running the same analysis on the whole MM08 product line though takes 8,788ms, resulting in a 13,021KB database.  Fig.~\ref{fig:graph_inference} is a log-scale bar graph of the speedup factor and the DB space savings factor for \insens. Speedup and space savings trends are again proportional to the number of valid configurations. For example, for BerkeleyDB, lifted \insens~is about 7.4 billion times faster than brute-force, with a DB 5.6 billion times smaller. All three analyses show similar speedup and disk space savings trends. 

Recall that the theoretic bottleneck of the lifted inference algorithm (Algorithm~\ref{alg:inference_lifted}) is the satisfiability checks performed when conjoining two PCs.  Since propositional satisfiability is  \term{NP-complete}, we wanted to evaluate whether it is a bottleneck in practice.
% problem, so this bottleneck is worth evaluating. 
While SAT checks are not required to maintain correctness of the lifted inference algorithm,
we perform them in order to avoid generating spurious facts that do not 
exist in any product. %actually exist in any product.
An UNSAT presence condition denotes an empty set of products, but what about PCs 
denoting sets of invalid product configurations? The \term{Feature Model (FM)} of a product line specifies which product configurations are valid and which are not. If a fact belongs only to a set of configurations excluded by the FM, then this fact can be removed. Removing spurious facts saves DB space, but, more importantly, keeps the set of facts searched by the inference algorithm as small as possible, improving the overall performance. We study the impact of SAT checking and using the FM below.
%in the next two research questions.

%\subsubsection{SAT Overhead}
% Table generated by Excel2LaTeX from sheet 'SAT'
\begin{table*}[htbp]
  \caption{SAT vs. noSAT. Time in milliseconds, Inferred DB in KB.}
  \label{tbl:noSAT}%
  \vspace{-0.1in}
  \small
  \centering
    \begin{tabular}{lrrrrrrrrrrrr}
        \toprule
            \multicolumn{1}{r}{} & \multicolumn{4}{c}{\textit{\textbf{insens}}} & \multicolumn{4}{c}{\textit{\textbf{1Type+Heap}}} & \multicolumn{4}{c}{\textit{\textbf{Taint-1Call+Heap}}} \\
        \cmidrule(r){2-13}
            \multicolumn{1}{r}{} & \multicolumn{2}{c}{\textit{\textbf{Time(ms)}}} & \multicolumn{2}{c}{\textit{\textbf{DB(KB)}}} & \multicolumn{2}{c}{\textit{\textbf{Time(ms)}}} & \multicolumn{2}{c}{\textit{\textbf{DB(KB)}}} & \multicolumn{2}{c}{\textit{\textbf{Time(ms)}}} & \multicolumn{2}{c}{\textit{\textbf{DB(KB)}}} \\
        \cmidrule(r){2-13}
            \multicolumn{1}{r}{} & \multicolumn{1}{l}{\textbf{SPL}} & \multicolumn{1}{l}{\textbf{noSAT}} & \multicolumn{1}{l}{\textbf{SPL}} & \multicolumn{1}{l}{\textbf{noSAT}} & \multicolumn{1}{l}{\textbf{SPL}} & \multicolumn{1}{l}{\textbf{noSAT}} & \multicolumn{1}{l}{\textbf{SPL}} & \multicolumn{1}{l}{\textbf{noSAT}} & \multicolumn{1}{l}{\textbf{SPL}} & \multicolumn{1}{l}{\textbf{noSAT}} & \multicolumn{1}{l}{\textbf{SPL}} & \multicolumn{1}{l}{\textbf{noSAT}} \\
        \midrule
            \textbf{BerkeleyDB} & 10,810 & 10,879 & 49,725 & 53,151 & 17,273 & 17,784 & 122,922 & 143,720 & 21,474 & 21,920 & 112,060 & 128,427 \\
        
            \textbf{GPL} & 4,517 & 4,496 & 409   & 472   & 4,718 & 4,667 & 593   & 765   & 8,861 & 8,812 & 462   & 595 \\
        
            \textbf{Lampiro} & 8,111 & 8,105 & 41,170 & 41,170 & 20,725 & 21,224 & 149,686 & 149,710 & 45,996 & 47,132 & 230,370 & 230,370 \\
        
            \textbf{MM08} & 4,720 & 4,689 & 3,259 & 3,762 & 5,142 & 5,118 & 6,990 & 8,775 & 9,306 & 9,270 & 7,829 & 9,082 \\
        
            \textbf{Prevayler} & 5,334 & 5,050 & 4,407 & 4,638 & 6,013 & 5,940 & 8,630 & 9,392 & 9,717 & 9,903 & 5,534 & 5,869 \\
        \bottomrule
    \end{tabular}%
\end{table*}%

\begin{figure}[t]
    \includegraphics[width=0.5\textwidth]{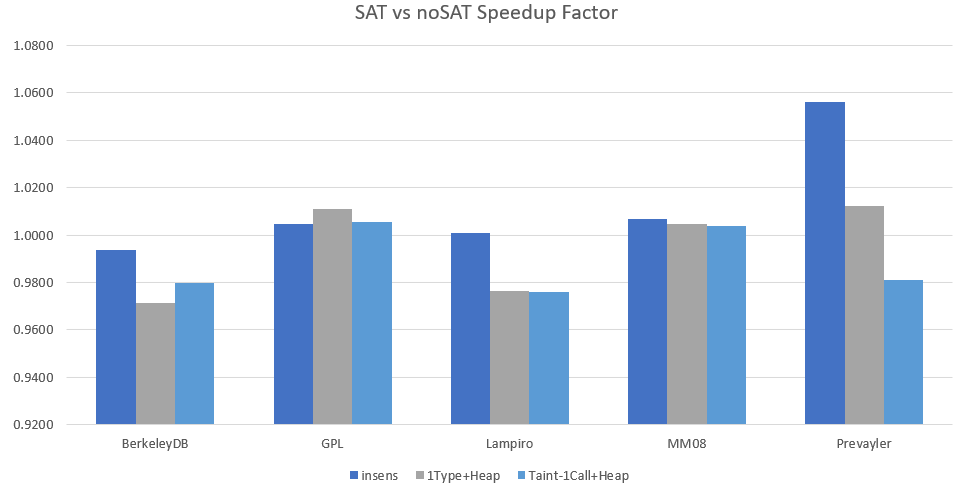}
    \caption{Inference time: SPL vs. SPL with SAT checking disabled for all three analyses.}
    \label{fig:graph_noSATTime}
    \vspace{-0.1in}
\end{figure}
\begin{figure}[t]
    \includegraphics[width=0.5\textwidth]{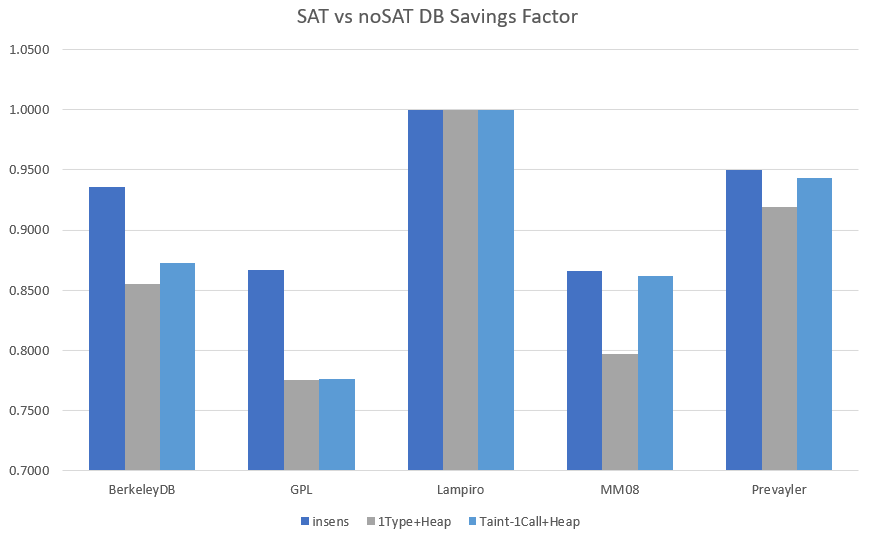}
    \caption{Inferred database size (KB): SPL vs. SPL with SAT checking disabled for all three analyses.}
    \label{fig:graph_noSATSpace}
    \vspace{-0.1in}
\end{figure}

\RQ{2.1}: How much does SAT checking contribute to the processing time of the lifted Datalog engine?

Table~\ref{tbl:noSAT} summarizes the
performance of our lifted analyses and the same analyses with SAT checking disabled (noSAT). Fig.~\ref{fig:graph_noSATTime} and Fig.~\ref{fig:graph_noSATSpace} show the noSAT-associated speedup and 
database size savings,
respectively. 
Recall that we represent PCs using BDDs.
SAT checking over BDDs is a constant-time operation~\cite{Huth:2004}.
Since conjoining and disjoining BDDs can take exponential time, we disable all BDD operations, keeping only the textual representation of PCs.
%{\bf MC:  I am confused.  Maybe describe this more?}
%is the potentially exponential time operation, so we actually disable all BDD operations.
A speedup factor below 1.0 means that disabling SAT checks slows down inference. This is
what we observed for most of the benchmarks.  We believe that the slowdown is because
of the use of textual representation of PCs which resulted in a much bigger PC table, with slower
lookup times.
%This is actually the case for most data points we observe. The reason behind that theoretically unexpected slowdown is that instead of storing presence conditions as BDDs in our presence condition table, we now store textual representations of presence conditions. A given proposition has one unique BDD representation, while it can be textually written in lots of different ways. As a result, using textual representations of presence conditions instead of BDDs results in a much bigger table, with slower lookup time. 
We also do not see any DB savings because non-canonically represented PCs tend to be longer than BDD-based
ones, resulting, on average, in more characters (and bytes) per PC.
 %The reason is non-canonical presence conditions tend to be longer than BDD-based canonical ones, resulting in more characters (and bytes) per presence condition on average.
We note that the number of features is relatively low in all of our benchmarks. BDD-based SAT solving is known to perform well on such small number of propositional variables. With product lines of hundreds or thousands of features, it is possible that noSAT might result in performance improvements.

%\subsubsection{Using the Feature Model}
% Table generated by Excel2LaTeX from sheet 'FM'
\begin{table*}[htbp]
  \caption{SPL vs. SPL+FM. Time in milliseconds, Inferred DB in KB.}
  \label{tbl:FM}%
  \vspace{-0.1in}
  \small
  \centering
    \begin{tabular}{lrrrrrrrrrrrr}
    \toprule
        \multicolumn{1}{r}{} & \multicolumn{4}{c}{\textit{\textbf{insens}}} & \multicolumn{4}{c}{\textit{\textbf{1Type+Heap}}} & \multicolumn{4}{c}{\textit{\textbf{Taint-1Call+Heap}}} \\
    \cmidrule(r){2-13}
        \multicolumn{1}{r}{} & \multicolumn{2}{c}{\textit{\textbf{Time(ms)}}} & \multicolumn{2}{c}{\textit{\textbf{DB(KB)}}} & \multicolumn{2}{c}{\textit{\textbf{Time(ms)}}} & \multicolumn{2}{c}{\textit{\textbf{DB(KB)}}} & \multicolumn{2}{c}{\textit{\textbf{Time(ms)}}} & \multicolumn{2}{c}{\textit{\textbf{DB(KB)}}} \\
    \cmidrule(r){2-13}
        \multicolumn{1}{r}{} & \multicolumn{1}{l}{\textbf{SPL}} & \multicolumn{1}{l}{\textbf{SPL+FM}} & \multicolumn{1}{l}{\textbf{SPL}} & \multicolumn{1}{l}{\textbf{SPL+FM}} & \multicolumn{1}{l}{\textbf{SPL}} & \multicolumn{1}{l}{\textbf{SPL+FM}} & \multicolumn{1}{l}{\textbf{SPL}} & \multicolumn{1}{l}{\textbf{SPL+FM}} & \multicolumn{1}{l}{\textbf{SPL}} & \multicolumn{1}{l}{\textbf{SPL+FM}} & \multicolumn{1}{l}{\textbf{SPL}} & \multicolumn{1}{l}{\textbf{SPL+FM}} \\
    \midrule
        \textbf{BerkeleyDB} & 10,810 & 11,693 & 49,725 & 721,927 & 17,273 & 22,276 & 122,922 & 1,859,320 & 21,474 & 24,140 & 112,060 & 1,625,278 \\
    
        \textbf{GPL} & 4,517 & 4,587 & 409   & 9,447 & 4,718 & 4,728 & 593   & 13,809 & 8,861 & 8,918 & 462   & 10,644 \\
    
        \textbf{Lampiro} & 8,111 & 8,528 & 41,170 & 325,848 & 20,725 & 22,283 & 149,686 & 1,283,639 & 45,996 & 48,688 & 230,370 & 1,843,151 \\
    
        \textbf{MM08} & 4,720 & 4,761 & 3,259 & 69,017 & 5,142 & 5,288 & 6,990 & 158,732 & 9,306 & 9,476 & 7,829 & 158,330 \\
    
        \textbf{Prevayler} & 5,334 & 5,169 & 4,407 & 8,825 & 6,013 & 5,984 & 8,630 & 17,564 & 9,717 & 9,977 & 5,534 & 11,394 \\
    \bottomrule
    \end{tabular}%
\end{table*}%

\begin{figure}[t]
    \includegraphics[width=0.5\textwidth]{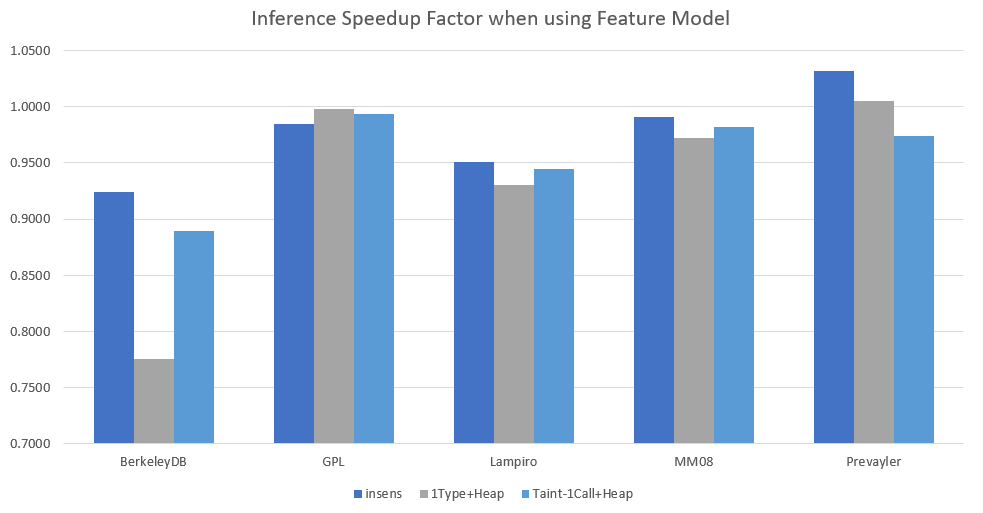}
    \caption{Inference time: SPL vs. SPL with FM for all three analyses.}
    \label{fig:graph_FMTime}
    \vspace{-0.1in}
\end{figure}
\begin{figure}[t]
    \includegraphics[width=0.5\textwidth]{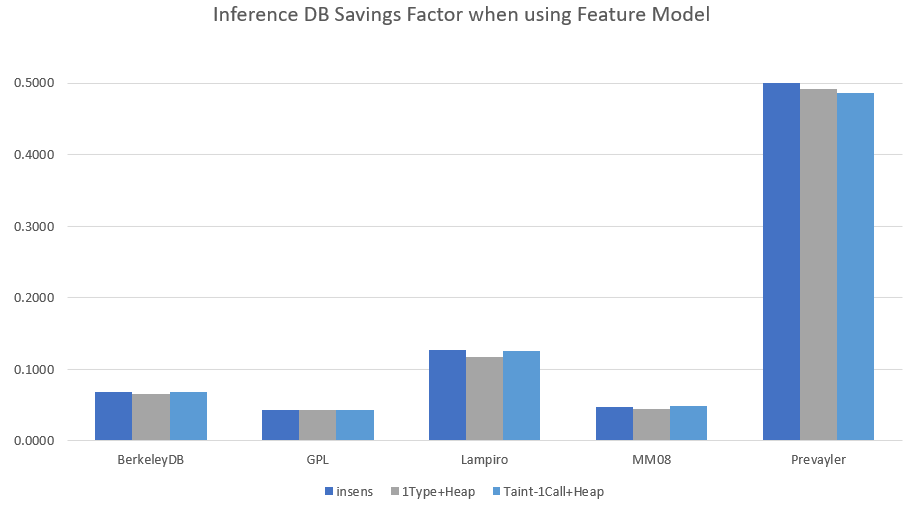}
    \caption{Inferred database size (KB): SPL vs. SPL with FM for all three analyses.}
    \label{fig:graph_FMSpace}
    \vspace{-0.1in}
\end{figure}
% Table generated by Excel2LaTeX from sheet 'Fact Count'
\begin{table}[htbp]
  \caption{The number of inferred facts with and without the Feature Model (FM).}
  \label{tbl:factCount}%
  \vspace{-0.1in}
  \footnotesize
  \centering
    \begin{tabular}{lrrrrrr}
    \toprule
          & \multicolumn{2}{c}{\textit{\textbf{insens}}} & \multicolumn{2}{c}{\textit{\textbf{1Type+Heap}}} & \multicolumn{2}{c}{\textit{\textbf{Taint-1Call+Heap}}} \\
    \cmidrule(r){2-7}
          & \multicolumn{1}{l}{\textbf{SPL}} & \multicolumn{1}{l}{\textbf{SPL+FM}} & \multicolumn{1}{l}{\textbf{SPL}} & \multicolumn{1}{l}{\textbf{SPL+FM}} & \multicolumn{1}{l}{\textbf{SPL}} & \multicolumn{1}{l}{\textbf{SPL+FM}} \\
    \midrule
    \textbf{BerkeleyDB} & 200,655 & 200,650 & 483,206 & 483,201 & 443,737 & 443,732 \\
    \textbf{GPL} & 3,528 & 3,128 & 5,675 & 4,821 & 4,450 & 3,800 \\
    \textbf{Lampiro} & 246,285 & 246,221 & 981,549 & 980,823 & 1,394,826 & 1,394,825 \\
    \textbf{MM08} & 13,021 & 13,021 & 29,453 & 29,452 & 29,383 & 29,381 \\
    \textbf{Prevayler} & 19,808 & 19,808 & 38,747 & 38,747 & 26,279 & 26,279 \\
    \bottomrule
    \end{tabular}%
    \vspace{-0.1in}
\end{table}%

\RQ{2.2}: What is the effect of taking the feature model (FM) of an SPL into consideration when running Datalog variability-aware analyses, in terms of inference time and DB size? 

Table~\ref{tbl:FM} compares the performance of our lifted analyses against the same analyses 
using the feature model (SAT+FM).  SAT+FM entails conjoining the feature model to each PC before
performing the satisfiability check. If the PC encodes a set of products excluded by the FM,
the conjunction is unsatisfiable.  Fig.~\ref{fig:graph_FMTime} and Fig.~\ref{fig:graph_FMSpace} show the SAT+FM-associated speedup and space savings, respectively.   For most of the experiments, using the FM results in slowdowns and larger DBs.  FM usage reduces the number of inferred facts, as
observed in Table~\ref{tbl:factCount}, but the reduction is relatively small.
 On the other hand, PCs now conjoined with the FM are more complex,  taking longer to construct (hence the performance penalty), and more bytes to store (hence the bigger DBs).

%\vspace{-0.1in}
\subsection{Threats to Validity}
%Our experiments have a number of internal and external threats to validity.
%
For internal threats, we note that all of our benchmarks are CIDE product lines.  While our lifting
approach and implementation are not specific to CIDE, CIDE limitations make the benchmarks biased towards specific annotation patterns. For example, only well-behaved annotations are allowed.  Furthermore, since feature expressions do not support feature negation, all input PCs are satisfiable, as well as conjunctions over those PCs. We experimented with disabling satisfiability checks to see how much they affect performance (while they always return true for this set of benchmarks). As noted previously, the overhead of those checks is marginal.

Another internal threat is that we approximate average product performance using only two samples (the maximum and the minimum). These averages are not expected to be completely accurate, but are used to give a brute-force estimate.  Our experiments show performance improvement of
several orders of magnitude, so we believe that our approximation (compared to more elaborate configuration sampling techniques) can be tolerated.

%As presented in our experimental results, speedups are several orders of magnitude, so lack of precision on average calculation can be tolerated.
 Finally,  all of the the analyses we used come from the Doop framework. Again, nothing in our lifted inference engine is Doop-specific, but extraction of annotated features is a part of Doop. Other frameworks can extract fact annotations in a similar fashion.
 %similarly in a straightforward fashion. 
%\end{enumerate}

\section{Related Work}
\label{sec:related}

Different kinds of software analyses have been re-implemented to support product lines~\cite{Thum:2014}. For example, the TypeChef project~\cite{Kastner:2011,Kastner:2012} implements variability aware parsers~\cite{Kastner:2011} and type checkers~\cite{Kastner:2012} for Java and C. The SuperC project~\cite{Gazzillo:2012} is another C language variability-aware parser. The Henshin~\cite{Arendt:2010} graph transformation engine was lifted to support product lines of graphs~\cite{Salay:2014}. Those lifted analyses were written from scratch, without reusing any components from their respective product-level analyses. Our approach, on the other hand, lifts an entire class of product-level analyses written as Datalog rules, by lifting their inference engine (and extracting presence conditions together with facts).

SPL\textsuperscript{Lift}~\cite{Bodden:2013} extends IFDS~\cite{Reps:1995} data flow analyses to product lines. Model checkers based on Featured Transition Systems~\cite{Classen:2013} check temporal properties of transition systems where transitions can be labeled by presence conditions. Both of these
SPL analyses  use almost the same single-product analyses on a lifted data representation. At a high level, our approach is similar in the sense that the logic of the original analysis is preserved, and only data is augmented with presence conditions. Still, our approach is unique because we do not touch any of the Datalog rules comprising the analysis logic itself.

Syntactic transformation techniques have been suggested for lifting abstract interpretation analyses to SPLs~\cite{Midtgaard:2015}. This line of work outlines  a systematic approach to lifting abstract interpretation analyses, together with correctness proofs.  Yet this approach is not automated  which means lifted analyses still need to be written from scratch, albeit while being guided by some systematic guidelines. 

Datalog engines have been used as backends by several program analysis frameworks. In addition to Doop, examples of analysis frameworks based on logic programming include XSB~\cite{Dawson:1996}, bddbddb~\cite{Whaley:2005} and Paddle~\cite{Lhotak:2008}. DIMPLE~\cite{Benton:2007} is another declarative pointer analysis framework where rules are written in Prolog. To the best of our knowledge, all those program analysis frameworks have been targeting single products. Our primary contribution is lifting this class of analyses to SPLs in a generic way, without making any analysis-specific assumptions. In addition, our approach can be systematically implemented in any Datalog engine used by any of those frameworks.

\section{Conclusion}
\label{sec:conclusion}
In this paper, 
presented an algorithm for lifting Datalog-based software analyses to SPLs.
We implemented this algorithm in the \souffle~Datalog engine, and evaluated performance
of three program analyses from the Doop framework on a suite of SPL benchmarks. Comparing our lifted implementation to brute-force analysis of each product individually, we show significant savings in terms of  processing time and database size.%, in the order of billions on one of the benchmarks.

Our \souffle~implementation only lifts the interpreter but not the code generator (compiler). Aggregation functions (e.g., sum, count) are not currently lifted either.  We plan to address these
implementation level limitations in future work.   We also plan to evaluate lifted \souffle~on analyses frameworks other than Doop. Another track for future work is lifting Datalog rules, not just facts. This would allow us to apply a product line of analyses to an SPL all at once. Our work can also be extended to lift Horn-Clause based analysis and verification tools~\cite{Bjrner:2015} to support SPLs.

\section* {Acknowledgments}
We thank Azadeh Farzan for discussions related to this work, and anonymous reviewers for their feedback on an earlier version of this paper. This work was supported by General Motors and NSERC.

\balance 
\bibliographystyle{ACM-Reference-Format}
\bibliography{datalog,spl} 
\end{document}